\documentclass{article}

\usepackage{arxiv}

\usepackage[utf8]{inputenc} 
\usepackage[T1]{fontenc}    
\usepackage{hyperref}       
\usepackage{url}            
\usepackage{booktabs}       
\usepackage{amsfonts}       
\usepackage{nicefrac}       
\usepackage{microtype}      
\usepackage{lipsum}
\usepackage{cite}
\usepackage{amssymb,latexsym,amsmath,graphicx,epstopdf,subfig}
\usepackage{float}

\newtheorem{thm}{Theorem}

\newtheorem{rem}{Remark}
\newtheorem{assumption}{Assumption}

\newtheorem{proof}{Proof}

\title{Discrete Adaptive Control Allocation}

\author{
  S. S.~Tohidi \\
  Mechasnical Engineering Department\\
  Bilkent University\\
  Ankara, 06800 Turkey \\
  \texttt{shahabaldin@bilkent.edu.tr} \\
   \And
 Y.~Yildiz \\
  Mechasnical Engineering Department\\
  Bilkent University\\
  Ankara, 06800 Turkey \\
  \texttt{yyildiz@bilkent.edu} \\
}

\begin{document}
\maketitle

\begin{abstract}
The main purpose of a control allocator is to distribute a total control effort among redundant actuators. This paper proposes a discrete adaptive control allocator for over-actuated sampled-data systems in the presence of actuator uncertainty. The proposed method does not require uncertainty estimation or persistency of excitation. Furthermore, the presented algorithm employs a closed loop reference model, which provides fast convergence without introducing excessive oscillations. To generate the total control signal, an LQR controller with reference tracking is used to guarantee the outer loop asymptotic stability. The discretized version of the Aerodata Model in Research Environment (ADMIRE) is used as an over-actuated system, to demonstrate the efficacy of the proposed method.
\end{abstract}


\section{Introduction}
Increasing the number of actuators in dynamical systems is one of the ways to improve maneuverability and fault tolerance \cite{TohYil20a, Toh16}. Thanks to the advances in microprocessors and progress in actuator miniaturization, which leads to actuator cost reduction, over-actuated systems are becoming ubiquitous in engineering applications. Aerial vehicles \cite{Duc09, SadChaZhaThe12, She15, She17, AcoYil14, Yil11a, TohYil18}, marine 
vehicles \cite{Che13, Joh08}, and automobiles \cite{TjoJoh10, Tem18, Tem20, Toh13} can be counted as examples of systems where redundant actuators are employed.

Allocating control signals among redundant actuators can be achieved via several different control allocation methods which can be categorized into the following categories: Pseudo-inverse-based, optimization-based and dynamic control allocation. Pseudo-inverse-based methods \cite{Alw08, Toh16}, which have the lowest computational complexity among the others, are implemented by manipulating the null space of the control input matrix. Optimization-based methods \cite{Bod02, Har05, YilKol10, YilKolAco11, Yang19} are performed by minimizing a cost function that penalizes the difference between the desired and achieved total control inputs. Dynamic control allocation methods \cite{FalHol16, Gal18, TohYil20a}, on the other hand, are based on solving differential equations that model the control allocation goals. These methods can be extended to consider actuator limits \cite{TohYil17, TohYil20b}. A survey of control allocation methods can be found in \cite{JohFos13}.

Actuator effectiveness uncertainty is an inevitable problem in several dynamical systems and it can be handled by control allocation methods in over-actuated systems. However, most of the control allocation methods that are proposed to solve this problem require uncertainty estimation and persistently exciting input signals \cite{CasGar10, Toh16}. Adaptive control allocation \cite{TjoJoh08, TohYil20a, TohYil16}, on the other hand, is able to manage redundancy and uncertainty of actuators without these requirements.

The majority of control algorithms are implemented using digital technology. Therefore, before their application, continuous-time controllers need an intermediate discretization step, which may lead to loss of stability margins \cite{San10, Dog20}. 

In this paper, a discrete adaptive control allocation method is introduced for sampled-data systems. The approach is inspired by the continuous control allocation algorithm that is recently proposed in \cite{TohYil20a}.  It does not require uncertainty estimation or persistency of excitation assumption. Furthermore, the method is implemented using a closed loop reference model \cite{GibAnnLav15}, which proved to speed up the system response without causing excessive oscillations \cite{Alan18}. To the best of our knowledge, a discrete adaptive control allocator with these properties is not available in the prior literature.

This paper is organized as follows: Section \ref{sec:2} introduces the notations and definitions employed during the paper. Section \ref{sec:3} presents
control allocation problem statement. The discrete adaptive control
allocation is presented in Section \ref{sec:4}. Controller design is presented in section
\ref{sec:5}. Section \ref{sec:6} illustrates the effectiveness of the proposed
method in the simulation environment. Finally, Section \ref{sec:7}
summarizes the paper.

\section{NOTATIONS}\label{sec:2}
In this section, we collect several definitions and basic results which are used in the following sections. Throughout this paper, $\lambda_{\text{min}}(.)$, $\lambda_{\text{max}}(.)$ and $\lambda_{i}(.)$ refer to the minimum, maximum and the $ i^{\text{th}} $ eigenvalue of a matrix, respectively. $ I_r $ is the identity matrix of dimension $ r\times r $, $ 0_{r\times m} $ is the zero matrix of dimension $ r\times m $. $ \text{tr}(.) $ refers to the trace operation and $ \text{diag}([.]) $ symbolizes a diagonal matrix with the elements of a vector $ [.] $. $ \mathbb{R} $, $ \mathbb{R}^n $ and $ \mathbb{R}^{n\times m} $ denote the set of real numbers, real column vectors with $ n $ elements, and $ n\times m $ real matrices, respectively. $ \mathbb{Z}^+ $ denotes the set of non-negative integers. In the discrete time, the $ \mathcal{L}_2 $ and $ \mathcal{L}_{\infty} $ signal norms are defined as
\begin{align}
||x(k)||=||x(k)||_2&=\sqrt{\sum_{k=0}^{\infty}\left(x_1^2(k)+ ... + x_n^2(k) \right) }, \\
||x(k)||_{\infty}&=\sup_{k\in \mathbb{Z}^+}\max_{1\leq i\leq n}|x_i(k)|,
\end{align}
where $ x_1, ..., x_n $ are the elements of $ x $. $ x(k)\in L_2 $ if $ ||x(k)||_2<\infty $. Also, $ x(k)\in L_{\infty} $ if $ ||x(k)||_{\infty}<\infty $.

\section{PROBLEM STATEMENT}\label{sec:3}


Consider the following discretized plant dynamics
\begin{align}\label{eq:1}
{x}(k+1)&=Ax(k)+B_uu(k)\notag\\
&=Ax(k)+B_vBu(k),
\end{align}
where $ k\in \mathbb{Z}^+ $ is the sampling instant, $x\in \mathbb{R}^{n}$ is the system states vector, $u\in \mathbb{R}^{m}$ is the control input vector, $A \in \mathbb{R}^{n\times n}$ is the known state matrix and $B_u \in \mathbb{R}^{n\times m}$ is the known control input matrix. 
\begin{rem}\label{rem0}
	It is noted that in the development of the control allocator, the matrix $ A $ being known or unknown is irrelevant. This matrix will play a role in the controller (not control allocator) design and since the contribution of the paper is a novel discrete time adaptive control allocator, the matrix $ A $ is taken to be known to facilitate the controller (not control allocator) development. 
\end{rem}

Redundancy of actuators leads $ B_u $ in (\ref{eq:1}) to be rank deficient, that is, $ \text{rank}(B_u)=r< m $.  Therefore, $ B_u $ can be decomposed into the known matrices $B_v \in \mathbb{R}^{n\times r}$ and $B \in \mathbb{R}^{r \times m}$ with $ \text{rank}(B_v)=\text{rank}(B)=r $. To model the actuator degradation, a diagonal matrix $\Lambda \in \mathbb{R}^{m\times m}$ with uncertain positive elements, belong to $ (0, 1] $, is introduced to the system dynamics as
\begin{align}\label{eq:2}
{x}(k+1)&=Ax(k)+B_vB\Lambda u(k)\notag \\
&=Ax(k)+B_vv(k),
\end{align}
where $v\in \mathbb{R}^{r}$ denotes the bounded control input produced by the controller. The boundedness of the control input can be guaranteed by using a soft saturation on the control signal $ v $, before feeding it to the control allocator. An example of this can also be seen in \cite{TohYil20a}.

The control allocation problem is to achieve
\begin{equation}\label{eq:3}
B\Lambda u(k)=v(k),
\end{equation} 
without using any matrix identification methods. Since $\Lambda$ is unknown, conventional control allocation methods do not apply.


\section{DISCRETE ADAPTIVE CONTROL ALLOCATION}\label{sec:4}

Consider the following dynamics
\begin{equation}\label{eq:4}
{\xi}(k+1)=A_m\xi(k)+B\Lambda u(k)-v(k),
\end{equation}
where $A_m \in \mathbb{R}^{r\times r}$  is stable matrix, that is, eigenvalues of $ A_m $ are inside the unit circle. A reference model dynamics is chosen as
\begin{equation}\label{eq:5}
{\xi}_m(k+1)=A_m\xi_m(k).
\end{equation}
Defining the control input as a mapping from $v$ to $u$,
\begin{equation}\label{eq:6}
u(k)={\theta}_v^T(k)v(k),
\end{equation}
where $\theta_v \in \mathbb{R}^{r\times m}$ represents the adaptive parameter matrix to be determined, and substituting (\ref{eq:6}) into (\ref{eq:4}), it is obtained that
\begin{equation}\label{eq:7}
{\xi}(k+1)=A_m\xi(k)+(B\Lambda{\theta}_v^T(k)-I_r)v(k).
\end{equation}

\hfill \break
Defining $\theta_v(k)=\theta_v^*+\tilde{\theta}_v(k)$, where ${\theta_v^*}=((B\Lambda)^T(B\Lambda \Lambda B^T)^{-1})^T$ is the ideal value of $\theta_v$, which corresponds to the pseudo inverse of $B\Lambda$, and $\tilde{\theta}_v$ is the deviation of $\theta_v$ from its ideal value, equation (\ref{eq:7}) can be rewritten as
\begin{equation}\label{eq:8}
\xi(k+1)=A_m\xi(k)+B\Lambda \tilde{\theta}_v^T(k)v(k).
\end{equation}
Defining the error $e(k)=\xi(k)-\xi_m(k)$, and using (\ref{eq:5}) and (\ref{eq:8}), the error dynamics is obtained as
\begin{equation}\label{eq:9}
{e}(k+1)=A_me(k)+B\Lambda \tilde{\theta}_v^T(k)v(k).
\end{equation}
\begin{assumption}\label{Assum1}
	The design matrix $ A_m $ is chosen such that \cite{Ioa06}:
	
	\noindent
	(i) $ |\lambda_i(A_m)|\leq 1, i=1, ..., r $,
	
	\noindent
	(ii) All controllable modes of $ (A_m, B\Lambda) $ are inside the unit circle,
	
	\noindent
	(iii) The eigenvalues of $ A_m $ on the unit circle have a Jordan block of size one.
\end{assumption}
\begin{thm}\label{thm2}
	Consider the system $ x(k+1)=\hat{A}x(k)+\hat{B}u(k) $, which satisfies Assumption \ref{Assum1}. There exist positive constants $ m_1 $ and $ m_2 $, independent of $ k $ and $ N $, such that
	\begin{align}\label{eq:thm2}
	||x(k)||\leq m_1+m_2\max_{0\leq \tau \leq N}||u(\tau)||,
	\end{align}
	for all $ k $, such that $ 0\leq k\leq N $.
\end{thm}
\begin{proof}
	\normalfont The proof can be found in \cite{Ioa06}.	$\blacksquare$
\end{proof}
\begin{thm}\label{thm3}
	Consider the error dynamics (\ref{eq:9}), which satisfies Assumption \ref{Assum1}. If the update law
	\begin{align}\label{eq:thm3}
	{\theta}_v(k+1)={\theta}_v(k)+\Gamma v(k)\epsilon^T(k)B
	\end{align}
	is used, where $ 0<\Gamma=\Gamma^T\in \mathbb{R}^{r\times r} $ is the adaptation rate matrix, and $ \epsilon(k)\in \mathbb{R}^r $ is defined as
	\begin{align}\label{eq:thm3x}
	\epsilon(k)=\frac{v(k)-B\Lambda u(k)}{\sigma^2(k)},
	\end{align}
	with $ \sigma(k)\equiv\sqrt{1+\lambda_{\text{max}}(B\Lambda B^T)v^T(k)\Gamma v(k)} $, then the adaptive parameter $ \theta_v(k) $, the error signal $ e(k) $ and all signals remain bounded. Furthermore, $ \lim_{k\rightarrow \infty}e(k)=0 $.
\end{thm}
\begin{proof}
\normalfont 	Consider the scalar positive definite function
	\begin{equation}\label{eq:10}
	V(k)=\text{tr}\left\lbrace \tilde{\theta}_v^T(k)\Gamma^{-1}\tilde{\theta}_v(k)\Lambda\right\rbrace ,
	\end{equation}
	where $ \Gamma=\Gamma^T>0 $. The time increment of (\ref{eq:10}) can be calculated as
	\begin{align}\label{eq:11}
	V(k+1)-V(k)&=\text{tr}\left\lbrace \tilde{\theta}_v^T(k+1)\Gamma^{-1}\tilde{\theta}_v(k+1)\Lambda \right\rbrace -\text{tr} \left\lbrace  \tilde{\theta}_v^T(k)\Gamma^{-1}\tilde{\theta}_v(k)\Lambda \right\rbrace.
	\end{align}
	Using (\ref{eq:thm3}) and the fact that $\theta_v(k)=\theta_v^*+\tilde{\theta}_v(k)$, it is obtained that
	\begin{align}\label{eq:11x}
	\tilde{\theta}_v(k+1)&=\theta_v(k+1)-\theta_v^*\notag \\ 
	&=\theta_v(k)+\Gamma v(k)\epsilon^T(k)B-\theta_v^*\notag \\
	&=\tilde{\theta}_v(k)+\Gamma v(k)\epsilon^T(k)B.
	\end{align}
	Substituting (\ref{eq:11x}) in (\ref{eq:11}), and using the trace property, $\text{tr}\left\lbrace A+B \right\rbrace=\text{tr}\left\lbrace A \right\rbrace + \text{tr}\left\lbrace B \right\rbrace   $ for two square matrices $ A $ and $ B $, we have
	\begin{align}\label{eq:13}
	V(k+1)-V(k)&=\text{tr}\Biggl\{ \tilde{\theta}_v^T(k)\Gamma^{-1}\tilde{\theta}_v(k)\Lambda+\tilde{\theta}_v^T(k)v(k)\epsilon^T(k)B\Lambda+B^T\epsilon(k)v^T(k)\tilde{\theta}_v(k)\Lambda\notag \\
	&+B^T\epsilon(k)v^T(k)\Gamma v(k)\epsilon^T(k)B\Lambda\Bigg\}-\text{tr}\left\lbrace \tilde{\theta}_v^T(k)\Gamma^{-1}\tilde{\theta}_v(k)\Lambda \right\rbrace\notag \\
	&=\text{tr}\Biggl\{ 2B^T\epsilon(k)v^T(k)\tilde{\theta}_v(k)\Lambda+B^T\epsilon(k)v^T(k)\Gamma v(k)\epsilon^T(k)B\Lambda\Bigg\}.
	\end{align}
	Since $ u(k)=\theta_v^{T}v(k) $ and $ v(k)=B\Lambda \theta_v^{*^T}v(k) $, (\ref{eq:thm3x}) can be rewritten as
	\begin{align}\label{eq:thm3xx}
	\epsilon(k)=\frac{-B\Lambda \tilde{\theta}_v^T(k)v(k)}{\sigma^2(k)}.
	\end{align}
	By substituting (\ref{eq:thm3xx}) in (\ref{eq:13}) and using the trace properties, $ \text{tr}\left\lbrace cA \right\rbrace =c\times  \text{tr}\left\lbrace A \right\rbrace   $, for a square matrix $ A $ and a scalar $ c $, and $ a^Tb=\text{tr}\left\lbrace ba^T \right\rbrace  $, for two column vectors $ a $ and $ b $, (\ref{eq:13}) can be rewritten as
	\begin{align}\label{eq:14}
	V(k+1)-V(k)&=\text{tr}\Biggl\{ \frac{-2B^TB\Lambda \tilde{\theta}_v^T(k)v(k)v^T(k)\tilde{\theta}_v(k)\Lambda}{\sigma^2(k)}+\frac{B^TB\Lambda \tilde{\theta}_v^T(k)v(k)v^T(k)\Gamma v(k)v^T(k)\tilde{\theta}_v(k)\Lambda B^TB\Lambda}{\sigma^4(k)} \Biggr\}\notag \\
	&=\frac{1}{\sigma^2}\text{tr}\Biggl\{ -2B^TB\Lambda \tilde{\theta}_v^T(k)v(k)v^T(k)\tilde{\theta}_v(k)\Lambda+\frac{B^TB\Lambda \tilde{\theta}_v^T(k)v(k)v^T(k)\Gamma v(k)v^T(k)\tilde{\theta}_v(k)\Lambda B^TB\Lambda}{\sigma^2(k)} \Biggr\}\notag \\
	&=\frac{1}{\sigma^2(k)}\text{tr}\Biggl\{ -2v^T(k)\tilde{\theta}_v(k)\Lambda B^TB\Lambda \tilde{\theta}_v^T(k)v(k)\notag \\
	&+ \frac{v^T(k)\Gamma v(k)v^T(k)\tilde{\theta}_v(k)\Lambda B^TB\Lambda B^TB\Lambda \tilde{\theta}_v^T(k)v(k)}{\sigma^2(k)} \Biggr\}.
	\end{align}
	Using the inequality $ a^TAa\leq \lambda_{\text{max}}(A)a^Ta $ for a symmetric matrix $ A $ and a column vector $ a $, an upper bound for (\ref{eq:14}) can be obtained as
	\begin{align}\label{eq:14x}
	V(k+1)-V(k) &\leq \frac{1}{\sigma^2(k)}\text{tr}\Biggl\{ -2v^T(k)\tilde{\theta}_v(k)\Lambda B^TB\Lambda \tilde{\theta}_v^T(k)v(k)\notag \\
	&+\lambda_{\text{max}}(B\Lambda B^T) \frac{v^T(k)\Gamma v(k)}{\sigma^2(k)}v^T(k)\tilde{\theta}_v(k)\Lambda B^TB\Lambda \tilde{\theta}_v^T(k)v(k) \Biggr\}\notag \\
	&=\frac{1}{\sigma^2(k)}\text{tr}\Biggl\{ v^T(k)\tilde{\theta}_v(k)\Lambda B^TB\Lambda \tilde{\theta}_v^T(k)v(k)\times \left( -2+\lambda_{\text{max}}(B\Lambda B^T)\frac{v^T(k)\Gamma v(k)}{\sigma^2(k)}\right)  \Biggr\}.
	\end{align}
	Considering the definition of $ \sigma(k) $, which is given after (\ref{eq:thm3x}), it can be obtained that $ -2\leq \left( -2+\lambda_{\text{max}}(B\Lambda B^T)\frac{v^T(k)\Gamma v(k)}{\sigma^2(k)} \right)<-1 $. Therefore, an upper bound for (\ref{eq:14x}) can be written as
	\begin{align}\label{eq:15}
	V(k+1)-V(k)< \frac{-1}{\sigma^2(k)} v^T(k)\tilde{\theta}_v(k)\Lambda B^TB\Lambda \tilde{\theta}_v^T(k)v(k)\leq 0.
	\end{align}
	This shows that $ V(k)\in \mathcal{L}_{\infty} $ and therefore $ \tilde{\theta}_v(k)\in \mathcal{L}_{\infty} $, which implies that $ {\theta}_v(k)\in \mathcal{L}_{\infty} $. In addition, since $ V(k) $ is decreasing and positive definite, it has a limit as $ k\rightarrow \infty $, that is, $ \lim_{k\rightarrow \infty}V(k)=V_{\infty} $. Furthermore, using Theorem \ref{thm2}, the error dynamics (\ref{eq:9}), and the boundedness of $ \tilde{\theta}_v $ and $ v $, it can be shown that $ e(t)\in \mathcal{L}_{\infty} $. Finally, since $ \xi_m $ in (\ref{eq:5}) is bounded, $ \xi $ is also bounded.
	
	Summing both sides of (\ref{eq:15}) from $ k=0 $ to $ \infty $, it is obtained that
	\begin{align}\label{eq:17}
	& \sum_{k=0}^{\infty}\frac{1}{\sigma^2(k)}  \left( v^T(k)\tilde{\theta}_v(k)\Lambda B^TB\Lambda \tilde{\theta}_v^T(k)v(k)\right) \leq V(0)-V_{\infty}\leq \infty,\notag \\
	&\Rightarrow \lambda_{\text{min}}(\Lambda B^TB\Lambda )\sum_{k=0}^{\infty}\frac{1}{\sigma^2(k)}  \left( v^T(k)\tilde{\theta}_v(k)\tilde{\theta}_v^T(k)v(k)\right)\leq V(0)-V_{\infty}\leq \infty,\notag \\
	&\Rightarrow \sum_{k=0}^{\infty}\frac{v^T(k)\tilde{\theta}_v(k)\tilde{\theta}_v^T(k)v(k)}{\sigma^2(k)}\leq \frac{V(0)-V_{\infty}}{\lambda_{\text{min}}(\Lambda B^TB\Lambda )}\leq \infty. 
	\end{align}
	This leads to the conclusion that $ \frac{\tilde{\theta}_v^T(k)v(k)}{\sigma(k)}\in \mathcal{L}_2 $. Therefore \cite{Tao03}, 
	\begin{align}\label{eq:18}
	\lim_{k\rightarrow \infty}\frac{\tilde{\theta}_v^T(k)v(k)}{\sigma(k)}=0.
	\end{align}
	Since $ v(k) $ is bounded,  $ \sigma(k) $ is also bounded. Furthermore, $ \sigma(k)\geq 1 $ by definition. Therefore, using (\ref{eq:18}) it is obtained that 
	\begin{align}\label{eq:18x}
	\lim_{k\rightarrow \infty}\tilde{\theta}_v^T(k)v(k)=0_{m\times 1}.
	\end{align}
	Using (\ref{eq:18x}) and (\ref{eq:thm3xx}), it can be concluded that $ \epsilon(k) $ converges to zero. Considering (\ref{eq:thm3}) and the convergence of $ \epsilon(k) $ to zero, we deduce that $ \theta_v(k) $ converges to a constant value as $ k\rightarrow \infty $. Finally, (\ref{eq:18x}) and the error dynamics (\ref{eq:9}) lead to the conclusion that $ \lim_{k\rightarrow \infty}e(k)=0 $.$\blacksquare$
\end{proof}

\begin{rem}\label{rem1}
	To realize (\ref{eq:4}), the signal $ B\Lambda u(k) $ is required. In motion control applications, this signal corresponds to the net external forces and moments \cite{Har03}, which can be obtained via an inertial measurement unit (IMU). Examples of measuring/estimating this signal, without introducing delay or noise amplification, and employing it in real applications can be found in \cite{Kut07} and \cite{SieMud10}.
\end{rem}

\begin{rem}\label{rem2}
	To calculate $ \sigma $, whose definition is given after (\ref{eq:thm3x}), $ \lambda_{\text{max}}(B \Lambda B^T) $ needs to be computed. Although $ \Lambda $ is an unknown matrix, the range of its elements is known, which is $ (0,1] $. Therefore, the maximum eigenvalue of the matrix multiplication $ B \Lambda B^T $ can be calculated.
\end{rem}

To obtain fast convergence without introducing excessive oscillations, the open loop reference model (\ref{eq:5}) is modified as a closed loop reference model \cite{GibAnnLav15} as follows
\begin{align}
{\xi}_{m1}(k+1)&=A_m\xi_{m1}(k)-\l\left( \xi(k)-\xi_{m1}(k)\right) \label{eq:21x} \\
&=A_m^c\xi_{m1}(k)-\l \xi(k),\label{eq:21xx}
\end{align}
where $ A_m^c\equiv A_m+\l I_r $ and $l$ is a scalar design parameter. Defining the error $e_1(k)=\xi(k)-\xi_{m1}(k)$, and using (\ref{eq:21x}) and (\ref{eq:8}), the error dynamics is obtained as
\begin{align}\label{eq:22x}
{e}_1(k+1)&=A_me_1(k)+B\Lambda \tilde{\theta}_v^T(k)v(k)-\l e_1(k)\notag \\
&=\bar{A}_me_1(k)+B\Lambda \tilde{\theta}_v^T(k)v(k),
\end{align}
where $ \bar{A}_m\equiv A_m-\l I_r $. It is noted that the design parameter $ \l $ needs to be chosen such that both $ A_m^c $ in (\ref{eq:21xx}) and $ \bar{A}_m $ in (\ref{eq:22x}) satisfy Assumption \ref{Assum1}. 
\begin{thm}\label{thm4}
	Consider the error dynamics (\ref{eq:22x}), with the update law (\ref{eq:thm3}). The adaptive parameter $ \theta_v(k) $, error signal $ e_1(k) $ and all other signals remain bounded, and $ \lim_{k\rightarrow \infty}e_1(k)=0 $.
\end{thm}
\begin{proof}
\normalfont	The proof is similar to that of Theorem \ref{thm3}'s and therefore omitted here for brevity.$\blacksquare$
\end{proof}
\section{CONTROLLER DESIGN}\label{sec:5}
In this section, a discrete Linear Quadratic Regulator (LQR) is designed to generate the total control input for reference tracking. 

Consider the discrete dynamical system given in (\ref{eq:2}), together with the output vector $ y\in \mathbb{R}^{r} $ as
\begin{align}\label{eq:23}
x(k+1)&=Ax(k)+B_vv(k),\notag \\
y(k)&=Cx(k),
\end{align} 
where $ C\in \mathbb{R}^{r\times n} $ is the output matrix.
In order to design a controller for reference tracking, we define a new state as
\begin{align}\label{eq:24}
x_{new}(k+1)=x_{new}(k)+\Delta t(ref(k)-y(k)),
\end{align}
where $ x_{new}\in \mathbb{R}^r $ is created by integrating the tracking error, $ ref\in \mathbb{R}^r $ is the reference input and $ \Delta t\in \mathbb{R}^+ $ is the sampling interval. Augmenting (\ref{eq:24}) with (\ref{eq:23}), we have
\begin{align}
\begin{bmatrix}
x(k+1)\\ x_{new}(k+1)
\end{bmatrix}&=\begin{bmatrix}
A & 0_{n\times r} \\ -\Delta tC & I_r
\end{bmatrix}\begin{bmatrix}
x(k)\\ x_{new}(k)
\end{bmatrix}+\begin{bmatrix}
B_v\\ 0_{r\times r}
\end{bmatrix}v(k)
+\begin{bmatrix}
0_{n\times r}\\ \Delta tI_r
\end{bmatrix}ref(k).
\end{align}

Defining $ z(k)\equiv \left[ x^T(k)\ x_{new}^T(k) \right]^T\in \mathbb{R}^{n+r} $ as the state vector of the aggregate system, the cost function required to solve the LQR problem becomes
\begin{align}\label{eq:26}
J=\sum_{k=0}^{\infty}\left( z^T(k)Qz(k)+v^T(k)Rv(k) \right),
\end{align}
where $ Q\in \mathbb{R}^{(n+r)\times (n+r)} $ is a positive semi-definite matrix and $ R\in \mathbb{R}^{r\times r} $ is a positive definite matrix. 

Using (\ref{eq:26}) and solving the standard LQR problem for discrete systems lead to the optimal gain matrix $ K\in \mathbb{R}^{r\times (n+r)} $ and the control signal is obtained as $ v(k)=-Kz(k) $. The structure of the controller is shown in Figure \ref{fig:block}.
\begin{figure}
	\vspace{0.2cm}
	\centerline{\includegraphics[scale=0.5]{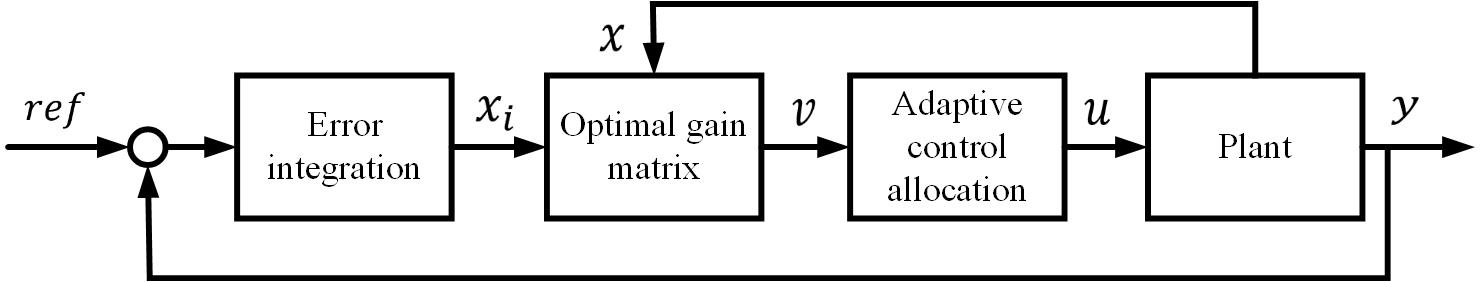}}
	\caption{Block diagram of the closed loop system. \label{fig:block}}
\end{figure}

\section{SIMULATION RESULTS}\label{sec:6}
\subsection{ADMIRE model}
The ADMIRE model is an over-actuated aircraft model introduced in \cite{Bat07} and \cite{Har03}. We use a version of this model that is linearized at Mach 0.22 and altitude 3km.
Considering the actuator loss of effectiveness matrix $ \Lambda $, and discretizing the continuous model using a $ 0.1 $s sampling interval, the discrete time dynamics is obtained as
\begin{align}\label{eq:e59}
x(k+1)&=Ax(k)+B_u u(k)\notag\\
&=Ax(k)+B_vBu(k) \notag \\
&=Ax(k)+B_vv(k),
\end{align}
where $ v=Bu $ is the control input, $ x=[\alpha \ \beta \ p \ q \ r]^T $ is the state matrix with $\alpha, \beta, p, q$ and $r$ denoting the angle of attack, sideslip angle, roll rate, pitch rate and yaw rate, respectively. The vector $ u=[u_c \ u_{re} \ u_{le} \ u_r]^T $ represents the control surface deflections of canard wings, right and left elevons and the rudder. The state and control matrices are given as
\begin{align}\label{eq:e62x}
A= \begin{bmatrix}1.0214 & 0.0054 & 0.0003 & 0.4176 & -0.0013\\
0  &  0.6307  &  0.0821   &   0  & -0.3792\\
0  & -3.4485  &  0.3979   &   0  &  1.1569\\
1.1199 & 0.0024 & 0.0001 &  1.0374 &  -0.0003\\
0  &  0.3802 &  -0.0156 &  0  &  0.8062
\end{bmatrix},
\end{align}
\begin{align}\label{eq:e62}
B_u= \begin{bmatrix}0.1823 & -0.1798 & -0.1795  &  0.0008\\
0 &  -0.0639 &   0.0639  &  0.1396 \\
0 & -1.584 & 1.584 & 0.2937 \\
0.8075 & -0.6456 &  -0.6456 &  0.0013\\
0 &  -0.1005  &  0.1005 &  -0.4113
\end{bmatrix}.
\end{align}
The uncertainty which is considered as actuator loss of effectiveness occurs at $ t=100 $s and reduces the actuator effectivenesses by $ 30\% $. It is noted that for the design of the control allocator, the first two rows of the matrix $ B_u $ is taken to be zero, which makes the control surfaces pure moment generators \cite{Har05}. However, the original $ B_u $ matrix given in (\ref{eq:e62}) is used for the plant dynamics in the simulations.

\subsection{Design parameters}
The design parameters of the controller are the matrices $ R $ and $ Q $, which are selected as $ R=\text{diag}([1,\ 1,\ 0.1]) $ and $ Q=I_8 $. Control allocation design parameters are $ \Gamma $ and $ A_m $. These two matrices are chosen as $ \Gamma=\text{diag}([1,\ 1,\ 0.1]) $ and $ A_m=\text{diag}([0.5,\ 0.5,\ 0.5]) $. To improve the transient response, the design parameter for the closed loop reference model approach is chosen as $ \l=0.1 $. 

\begin{figure}
	\centerline{\includegraphics[scale=0.5]{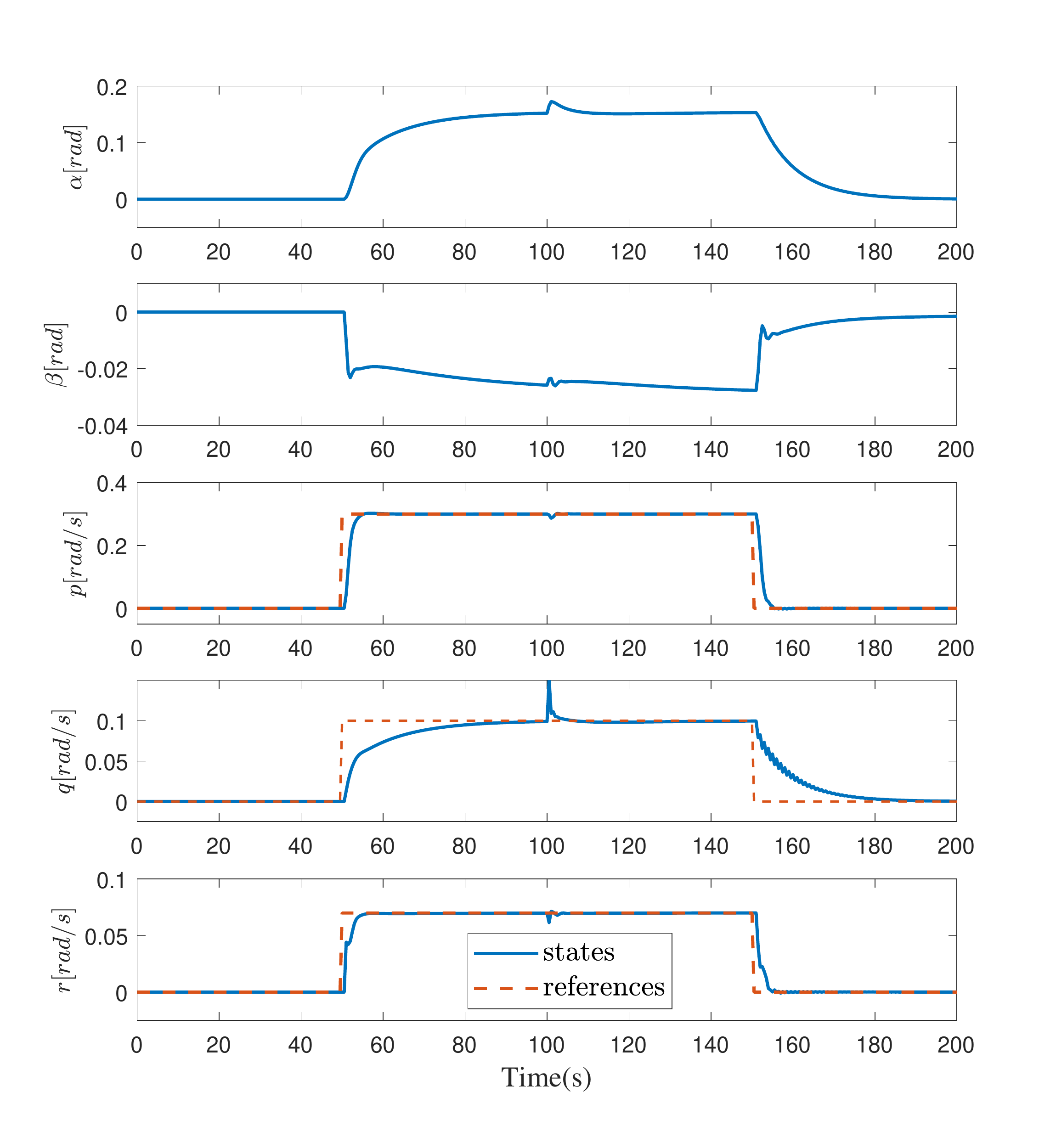}}
	\caption{System states and reference tracking.
		\label{fig:states}}
\end{figure}
\begin{figure}
	\centerline{\includegraphics[scale=0.5]{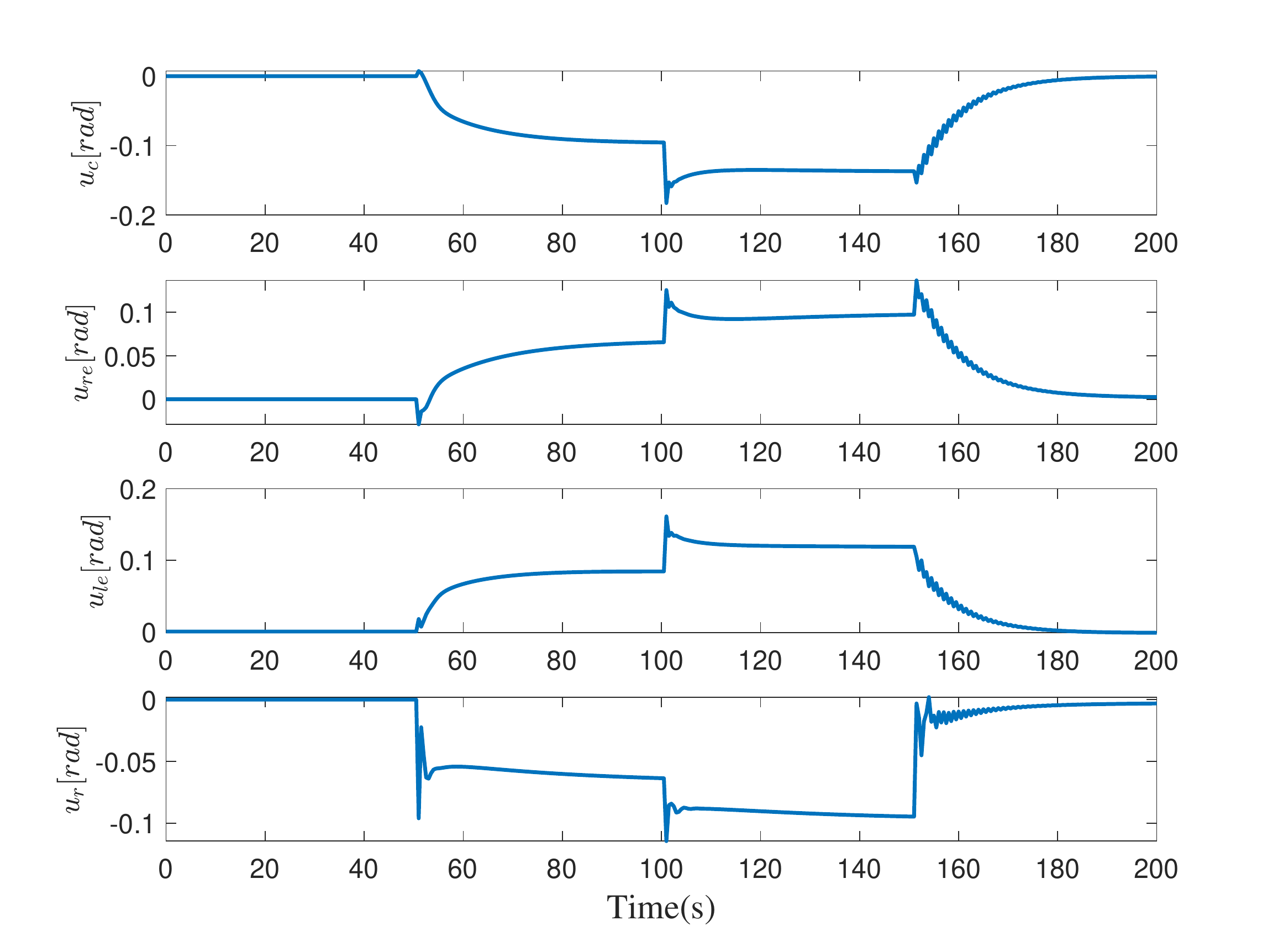}}
	\caption{Control surface deflections. \label{fig:u}}
\end{figure}
\begin{figure}
	\centerline{\includegraphics[scale=0.5]{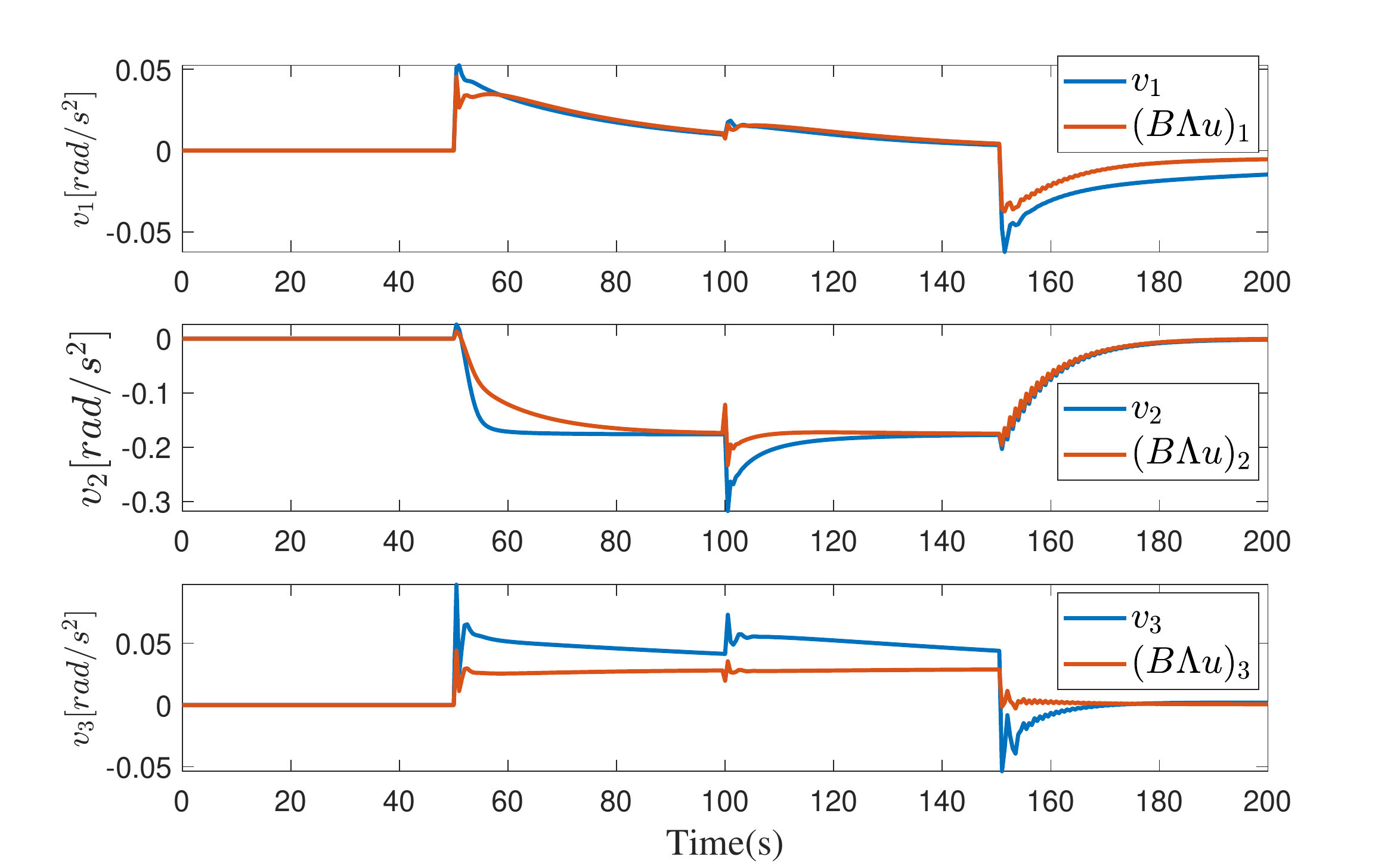}}
	\caption{Control input signal, $ v $, tracking. \label{fig:v}}
\end{figure}
\begin{figure}
	\centerline{\includegraphics[scale=0.6]{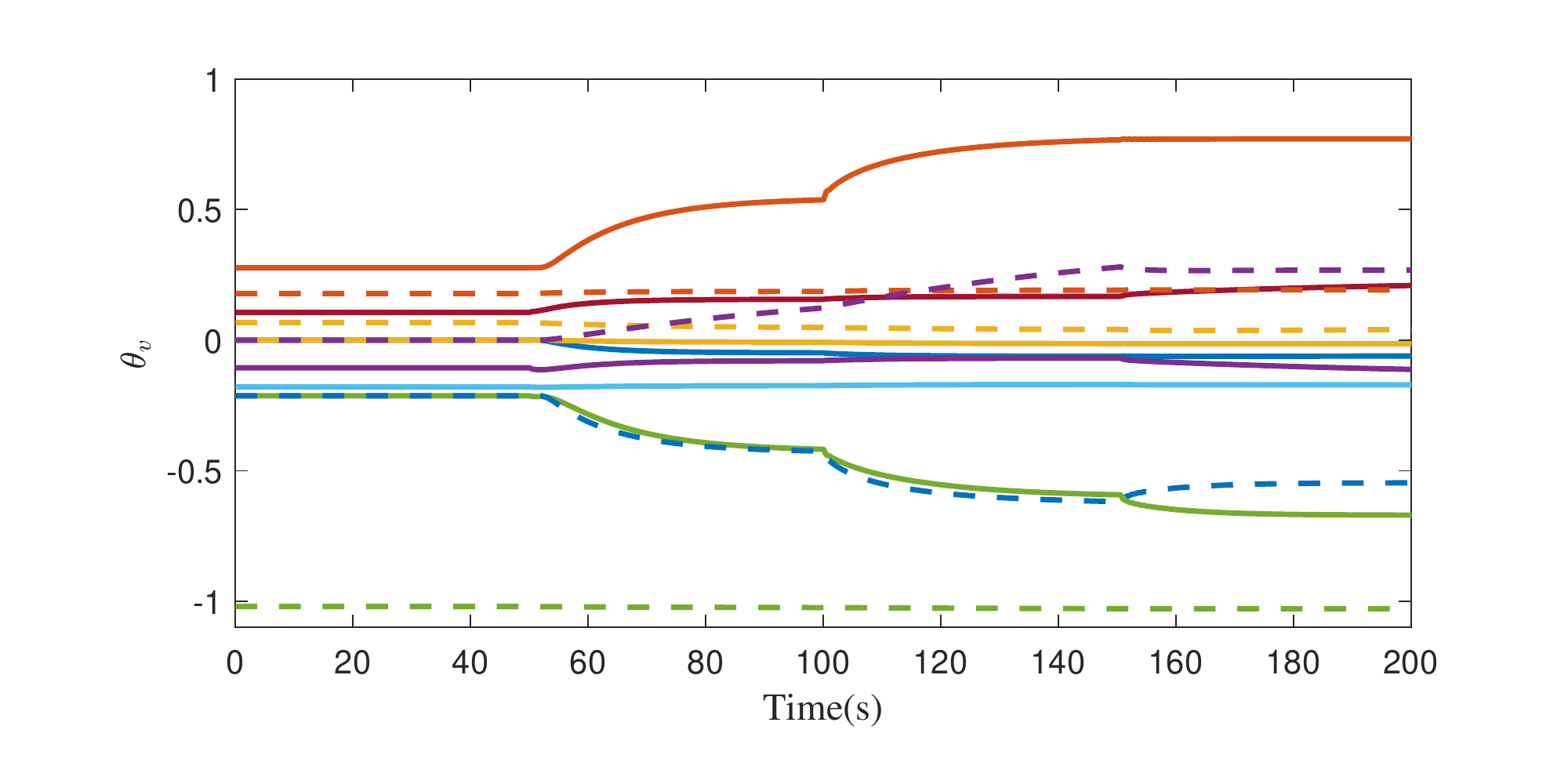}}
	\caption{Adaptive parameters. \label{fig:teta}}
\end{figure}

\subsection{Simulation results}
Figure \ref{fig:states} illustrates the aircraft states together with the three references. It is seen that the first two states ($ \alpha $ and $ \beta $) remain bounded while the other three states ($ p $, $ q $ and $ r $) track their references. The effect of the actuator effectiveness uncertainty, which is introduced at $ t=100 $s, can also be observed in this figure.   Figure \ref{fig:u} shows the time evolution of control surfaces, where no excessive deflections are observed. It is seen in Figure \ref{fig:v} that the total control signals, $ v_i, i=1,2,3 $,  are realized by the control allocator. The figure shows that $ B\Lambda u $ is converging to $ v $, which implies that the control allocation error is converging to zero. Adaptive parameters' time evolutions are demonstrated in Figure \ref{fig:teta}. The elements of $ \theta_v $ remain bounded throughout the simulation and eventually converge to constant values.

\section{Summary}\label{sec:7}
A discrete adaptive control allocation is proposed in this paper. This method is able to distribute the total control signals of a sampled-data system among redundant actuators in the presence of actuator effectiveness uncertainty.  The proposed control allocation method does not require uncertainty estimation or persistency of excitation. Simulation results demonstrate the effectiveness of the method.

\bibliographystyle{unsrt}  

%
%
\bibliography{References}

\begin{thebibliography}{10}

\bibitem{TohYil20a}
Seyed~Shahabaldin Tohidi, Yildiray Yildiz, and Ilya Kolmanovsky.
\newblock Adaptive control allocation for constrained systems.
\newblock {\em Automatica}, 121:109161, 2020.

\bibitem{Toh16}
Seyed~Shahabaldin Tohidi, Ali Khaki~Sedigh, and David Buzorgnia.
\newblock Fault tolerant control design using adaptive control allocation based
  on the pseudo inverse along the null space.
\newblock {\em International Journal of Robust and Nonlinear Control},
  26(16):3541--3557, 2016.

\bibitem{Duc09}
Guillaume~JJ Ducard.
\newblock {\em Fault-tolerant flight control and guidance systems: Practical
  methods for small unmanned aerial vehicles}.
\newblock Springer Science \& Business Media, 2009.

\bibitem{SadChaZhaThe12}
Iman Sadeghzadeh, Abbas Chamseddine, Youmin Zhang, and Didier Theilliol.
\newblock Control allocation and re-allocation for a modified quadrotor
  helicopter against actuator faults.
\newblock {\em IFAC Proceedings Volumes}, 45(20):247--252, 2012.

\bibitem{She15}
Qiang Shen, Danwei Wang, Senqiang Zhu, and Eng~Kee Poh.
\newblock Inertia-free fault-tolerant spacecraft attitude tracking using
  control allocation.
\newblock {\em Automatica}, 62:114--121, 2015.

\bibitem{She17}
Qiang Shen, Danwei Wang, Senqiang Zhu, and Eng~Kee Poh.
\newblock Robust control allocation for spacecraft attitude tracking under
  actuator faults.
\newblock {\em IEEE Transactions on Control Systems Technology},
  25(3):1068--1075, 2017.

\bibitem{AcoYil14}
Diana~M Acosta, Yildiray Yildiz, Robert~W Craun, Steven~D Beard, Michael~W
  Leonard, Gordon~H Hardy, and Michael Weinstein.
\newblock Piloted evaluation of a control allocation technique to recover from
  pilot-induced oscillations.
\newblock {\em Journal of Aircraft}, 52(1):130--140, 2014.

\bibitem{Yil11a}
Yildiray Yildiz and Ilya Kolmanovsky.
\newblock Stability properties and cross-coupling performance of the control
  allocation scheme capio.
\newblock {\em Journal of Guidance, Control, and Dynamics}, 34(4):1190--1196,
  2011.

\bibitem{TohYil18}
Seyed~Shahabaldin Tohidi, Yildiray Yildiz, and Ilya Kolmanovsky.
\newblock Pilot induced oscillation mitigation for unmanned aircraft systems:
  An adaptive control allocation approach.
\newblock In {\em IEEE Conference on Control Technology and Applications},
  pages 343--348, 2018.

\bibitem{Che13}
Mou Chen, Shuzhi~Sam Ge, Bernard Voon~Ee How, and Yoo~Sang Choo.
\newblock Robust adaptive position mooring control for marine vessels.
\newblock {\em IEEE Transactions on Control Systems Technology},
  21(2):395--409, 2013.

\bibitem{Joh08}
Tor~A Johansen, Thomas~P Fuglseth, Petter T{\o}ndel, and Thor~I Fossen.
\newblock Optimal constrained control allocation in marine surface vessels with
  rudders.
\newblock {\em Control Engineering Practice}, 16(4):457--464, 2008.

\bibitem{TjoJoh10}
Johannes Tj{\o}nn{\aa}s and Tor~A Johansen.
\newblock Stabilization of automotive vehicles using active steering and
  adaptive brake control allocation.
\newblock {\em IEEE Transactions on Control Systems Technology},
  18(3):545--558, 2010.

\bibitem{Tem18}
Ozan Temiz, Melih Cakmakci, and Yildiray Yildiz.
\newblock A fault tolerant vehicle stability control using adaptive control
  allocation.
\newblock In {\em Dynamic Systems and Control Conference}, volume 51890, page
  V001T09A002. American Society of Mechanical Engineers, 2018.

\bibitem{Tem20}
Ozan Temiz, Melih Cakmakci, and Yildiray Yildiz.
\newblock A fault-tolerant integrated vehicle stability control using adaptive
  control allocation.
\newblock {\em arXiv preprint arXiv:2008.05697}, 2020.

\bibitem{Toh13}
Seyed~Shahabaldin Tohidi and Ali~Khaki Sedigh.
\newblock Adaptive fault tolerance in automotive vehicle using control
  allocation based on the pseudo inverse along the null space for yaw
  stabilization.
\newblock In {\em The 3rd International Conference on Control, Instrumentation,
  and Automation}, pages 174--179. IEEE, 2013.

\bibitem{Alw08}
Halim Alwi and Christopher Edwards.
\newblock Fault tolerant control using sliding modes with on-line control
  allocation.
\newblock {\em Automatica}, 44(7):1859--1866, 2008.

\bibitem{Bod02}
Marc Bodson.
\newblock Evaluation of optimization methods for control allocation.
\newblock {\em Journal of Guidance, Control, and Dynamics}, 25(4):703--711,
  2002.

\bibitem{Har05}
Ola H{\"a}rkeg{\aa}rd and S.~Torkel Glad.
\newblock Resolving actuator redundancy-optimal control vs. control allocation.
\newblock {\em Automatica}, 41(1):137--144, 2005.

\bibitem{YilKol10}
Yildiray Yildiz and Ilya~V Kolmanovsky.
\newblock A control allocation technique to recover from pilot-induced
  oscillations (capio) due to actuator rate limiting.
\newblock In {\em American Control Conference}, pages 516--523, 2010.

\bibitem{YilKolAco11}
Yildiray Yildiz, Ilya~V Kolmanovsky, and Diana Acosta.
\newblock A control allocation system for automatic detection and compensation
  of phase shift due to actuator rate limiting.
\newblock In {\em American Control Conference}, pages 444--449, 2011.

\bibitem{Yang19}
Yuanchao Yang and Zichen Gao.
\newblock A new method for control allocation of aircraft flight control
  system.
\newblock {\em IEEE Transactions on Automatic Control}, 65(4):1413--1428, 2020.

\bibitem{FalHol16}
Guillermo~P Falcon{\'\i} and Florian Holzapfel.
\newblock Adaptive fault tolerant control allocation for a hexacopter system.
\newblock In {\em American Control Conference (ACC), 2016}, pages 6760--6766.
  IEEE, 2016.

\bibitem{Gal18}
Sergio Galeani and Mario Sassano.
\newblock Data-driven dynamic control allocation for uncertain redundant
  plants.
\newblock In {\em IEEE Conference on Decision and Control}, pages 5494--5499,
  2018.

\bibitem{TohYil17}
Seyed~Shahabaldin Tohidi, Yildiray Yildiz, and Ilya Kolmanovsky.
\newblock Adaptive control allocation for over-actuated systems with actuator
  saturation.
\newblock volume~50, pages 5492--5497. Elsevier, 2017.

\bibitem{TohYil20b}
Seyed~Shahabaldin Tohidi and Yildiray Yildiz.
\newblock Handling actuator magnitude and rate saturation in uncertain
  over-actuated systems: A modified projection algorithm approach.
\newblock {\em International Journal of Control}, pages 1--24, 2020.

\bibitem{JohFos13}
Tor~A Johansen and Thor~I Fossen.
\newblock Control allocation—a survey.
\newblock {\em Automatica}, 49(5):1087--1103, 2013.

\bibitem{CasGar10}
Alessandro Casavola and Emanuele Garone.
\newblock Fault-tolerant adaptive control allocation schemes for overactuated
  systems.
\newblock {\em International Journal of Robust and Nonlinear Control},
  20(17):1958--1980, 2010.

\bibitem{TjoJoh08}
Johannes Tj{\o}nn{\aa}s and Tor~A Johansen.
\newblock Adaptive control allocation.
\newblock {\em Automatica}, 44(11):2754--2765, 2008.

\bibitem{TohYil16}
Seyed~Shahabaldin Tohidi, Yildiray Yildiz, and Ilya Kolmanovsky.
\newblock Fault tolerant control for over-actuated systems: An adaptive
  correction approach.
\newblock In {\em American Control Conference, 2016}, pages 2530--2535, 2016.

\bibitem{San10}
Mario~A Santillo and Dennis~S Bernstein.
\newblock Adaptive control based on retrospective cost optimization.
\newblock {\em Journal of guidance, control, and dynamics}, 33(2):289--304,
  2010.

\bibitem{Dog20}
K~Merve Dogan, Tansel Yucelen, Wassim~M Haddad, and Jonathan~A Muse.
\newblock Improving transient performance of discrete-time model reference
  adaptive control architectures.
\newblock {\em International Journal of Adaptive Control and Signal
  Processing}, 2020.

\bibitem{GibAnnLav15}
Travis~E Gibson, Zheng Qu, Anuradha~M Annaswamy, and Eugene Lavretsky.
\newblock Adaptive output feedback based on closed-loop reference models.
\newblock {\em IEEE Transactions on Automatic Control}, 60(10):2728--2733,
  2015.

\bibitem{Alan18}
Anil Alan, Yildiray Yildiz, and Umit Poyraz.
\newblock High-performance adaptive pressure control in the presence of time
  delays: Pressure control for use in variable-thrust rocket development.
\newblock {\em IEEE Control Systems Magazine}, 38(5):26--52, 2018.

\bibitem{Ioa06}
Petros Ioannou and Bariş Fidan.
\newblock {\em Adaptive control tutorial}.
\newblock SIAM, 2006.

\bibitem{Tao03}
Gang Tao.
\newblock {\em Adaptive control design and analysis}, volume~37.
\newblock John Wiley \& Sons, 2003.

\bibitem{Har03}
Ola H{\"a}rkeg{\aa}rd.
\newblock {\em Backstepping and control allocation with applications to flight
  control}.
\newblock PhD thesis, Link{\"o}pings universitet, 2003.

\bibitem{Kut07}
Ali Kutay, John Culp, Jonathan Muse, Daniel Brzozowski, Ari Glezer, and Anthony
  Calise.
\newblock A closed-loop flight control experiment using active flow control
  actuators.
\newblock In {\em 45th AIAA Aerospace Sciences Meeting and Exhibit}, page 114,
  2007.

\bibitem{SieMud10}
S~Sieberling, QP~Chu, and JA~Mulder.
\newblock Robust flight control using incremental nonlinear dynamic inversion
  and angular acceleration prediction.
\newblock {\em Journal of guidance, control, and dynamics}, 33(6):1732--1742,
  2010.

\bibitem{Bat07}
Declan Bates and Martin Hagstr{\"o}m.
\newblock {\em Nonlinear analysis and synthesis techniques for aircraft
  control}, volume 365.
\newblock Springer, 2007.

\end{thebibliography}

\end{document}